\documentclass[12pt,reqno]{amsart}

\usepackage{times}
\usepackage{graphicx}
\graphicspath{ {./images/} }

\usepackage{romannum}

\usepackage{amsmath,amsthm,amsfonts,amscd,amssymb}
\usepackage{pst-node}
\usepackage{pst-plot}

\usepackage[all]{xy}

\usepackage{tikz, braids}
\usetikzlibrary{decorations.pathreplacing, arrows}

\newtheorem{theorem}{Theorem}[section]
\newtheorem{corollary}[theorem]{Corollary}

\newtheorem{remark}[theorem]{Remark}
\newtheorem{proposition}[theorem]{Proposition}

\numberwithin{theorem}{section}
\numberwithin{equation}{section}

\newcommand{\norm}[1]{\left\Vert#1\right\Vert}

\newcommand{\la}{\langle}
\newcommand{\ra}{\rangle}

\newcommand{\Comp}{\mathbb{C}}

\newcommand{\n}{\mathbb{N}}

\begin{document}

\title{Asymptotic analysis for $O_N^+$-Temperley-Lieb quantum channels}


\author{Sang-Gyun Youn}
\address{Sang-Gyun Youn, 
Department of Mathematics Education, Seoul National University, 
Gwanak-ro 1, Gwanak-gu, Seoul 08826, South Korea}
\email{s.youn@snu.ac.kr }

\thanks{This research was supported by National Research Foundation of Korea (NRF) grant funded by the Korea government (MSIT) (No. 2020R1C1C1A01009681) and by Samsung Science and Technology Foundation under Project Number SSTF-BA2002-01.}

\maketitle

\begin{abstract}

Studies on conservation of quantum symmetries have been initiated by recent papers \cite{BCLY20,LY20}. We, in this paper, focus on a class of quantum channels which are covariant for symmetries from free orthogonal quantum groups $O_N^+$. These quantum channels are called $O_N^+$-Temperley-Lieb channels, and their information-theoretic properties such as Holevo information and coherent information were analyzed in \cite{BCLY20}, but their additivity questions remained open. The main result of this paper is to approximate $O_N^+$-Temperley-Lieb quantum channels by much simpler ones in terms Bures distance. As applications, we study strong additivity questions for $O_N^+$-Temperley-Lieb quantum channels, and their classical capacity, private classical capacity and quantum capacity in the asymptotic regime $N\rightarrow \infty$.
\end{abstract}

\section{Introduction}

Conservation of (group) symmetry has been studied from various perspectives in quantum information theory (QIT) and there have been extensive efforts for so-called {\it invariant quantum states} and {\it covariant quantum channels}. Amongst them are \cite{Sc05,KW09,MS14,MHRW16,Ha17a,Ha17b,COS18} and, in particular, the covariance property with respect to compact group actions has been studied in \cite{VW01,DFH06,LS14,AN14,MSD17}.

A class of the simplest non-trivial covariant quantum channels is of the Werner-Holevo quantum channels, which is well-known as a couterexample on Amosov, Holevo and Werner's conjecture \cite{AHW00,WH02}. The Werner-Holevo channels are completely positive trace-preserving maps $\Phi:M_d(\Comp)\rightarrow M_d(\Comp)$ satisfying
\begin{equation}\label{eq1}
\Phi(U\rho U^*)=\overline{U}\Phi(\rho)U^t
\end{equation}
for all $U\in \mathcal{U}(d)$ and $ \rho \in M_d(\Comp)$. An important structure theorem for the Werner-Holevo channels is that they are generated by only two quantum channels, i.e. we have
\begin{align}\label{eq5}
\Phi(\rho)= \frac{1-p}{d+1}\left ( \text{Tr}(\rho)\text{Id}_d +\rho^t \right )+\frac{p}{d-1}\left ( \text{Tr}(\rho)\text{Id}_d - \rho^t \right )
\end{align}
for some $0\leq p\leq 1$. The given covariance property \eqref{eq1} is for the most fundamental unitary group actions $\rho\mapsto U\rho U^*$ and $\rho\mapsto \overline{U}\rho U^t$. A systematic approach to study such a structural analysis has been explored in \cite{VW01,Key02,AN14,MSD17,SC18}, and even for {\it quantum group actions} recently in \cite{BC18,BCLY20,LY20}. 

In particular, \cite{BC18,BCLY20} introduced the $O_N^+$-Temperley-Lieb (TL) quantum channels labeled by 
\begin{equation}
\Phi^{k\rightarrow l}_{N,|k-l|},\Phi^{k\rightarrow l}_{N,|k-l|+2},\cdots, \Phi^{k\rightarrow l}_{N,k+l},
\end{equation}
and \cite[Theorem 4.6 (4)]{LY20} implies that all irreducibly $O_N^+$-covariant quantum channels are convex combinations of those. Some important information-theoretic properties such as entanglement-breaking property, PPT property, (anti-)degradability, Holevo information and coherent information have been studied in \cite{BCLY20} for $O_N^+$-TL quantum channels in the asymptotic regime $N\rightarrow \infty$, whereas additivity questions including estimates for classical or quantum capacities were left as open questions. For example, it is shown in \cite[Corollary 4.6]{BCLY20} that
\begin{equation}
\lim_{N\rightarrow \infty} \left \{ \chi(\Phi^{k\rightarrow l}_{N,m}) -\frac{l+k-m}{2}\log(N)\right\}=0,
\end{equation}
but the question for the classical capacity $C(\Phi^{k\rightarrow l}_{N,m})$ was left open.

In this paper, we approximate Stinespring isometries of $O_N^+$-TL quantum channels to show that $O_N^+$-TL quantum channels are not far away from a class of well-analyzed quantum channels with respect to Bures distance. Then, combining the approximation and some of well-known continuity theorems for capacities, we can strengthen \cite[Corollary 4.6]{BCLY20} even for a tensorization with an arbitrary quantum channel $\Phi'$. Indeed, one of our main results is as follows:
\begin{equation}
\lim_{N\rightarrow \infty} \left \{ C(\Phi^{k\rightarrow l}_{N,m}\otimes \Phi') - C(\Phi^{k\rightarrow l}_{N,m}) - C(\Phi')\right\}=0
\end{equation}
and $\displaystyle \lim_{N\rightarrow \infty}\left \{C(\Phi^{k\rightarrow l}_{N,m}) - \displaystyle \frac{l+k-m}{2}\log(N)\right\}=0$. Moreover, the same conclusion holds for the private classical capacity $P$ and the quantum capacity $Q$.

\section{$O_N^+$-Temperley-Lieb quantum channels}

Recall that all irreducible unitary representations of the special unitary group $SU(2)$ are classified by $\pi_0,\pi_1,\pi_2,\cdots$ whose underlying Hilbert spaces are $H_0=\Comp, H_1=\Comp^2,H_2=\Comp^3,\cdots$ respectively. And the fusion rule in this representation category is given by
\begin{equation}
\pi_l\otimes \pi_m\cong \pi_{|l-m|}\oplus\pi_{|l-m|+2}\oplus \cdots \oplus \pi_{l+m}
\end{equation}
with a canonical Hilbert space decomposition 
\begin{equation}
H_l\otimes H_m=H_{|l-m|}\oplus H_{|l-m|+2}\oplus \cdots \oplus H_{l+m}. 
\end{equation}
The above representation-theoretic features can be captured by diagrammatic calculs from so-called {\it Temperley-Lieb category} (See \cite[Section 3.3]{BCLY20}) and such a unified approach covers a class of very important genuine quantum groups, namely {\it free orthogonal quantum groups} $O_N^+$ ($N\geq 2$). This is considered a universally quantized version of the orthogonal group $O_N$ \cite{Wa95,BS09}.

Indeed it is shown in \cite{Ba96} that all irreducible unitary representations of $O_N^+$ are classified by $u_0,u_1,u_2,\cdots$ up to unitary equivalence, and their underlying Hilbert spaces are denoted by $H_0,H_1,H_2,\cdots$. An important difference from the case of $SU(2)$ is that $\text{dim}(H_1)=N$ and $\text{dim}(H_k)=q^{k}+q^{k-2}+\cdots+q^{-k}$ where $\displaystyle q=\frac{2}{N+\sqrt{N^2-4}}$. Note that $\text{dim}(H_k) > k+1$ if $N\geq 3$. Despite such differences, the same the fusion holds for $O_N^+$, so the tensor product representation of $u_l$ and $u_m$ decomposes into
\begin{equation}
u_{|l-m|}\oplus u_{|l-m|+2}\oplus \cdots \oplus u_{l+m}
\end{equation}
with a canonical Hilbert space decomposition 
\begin{equation}
H_l\otimes H_m=H_{|l-m|}\oplus H_{|l-m|+2}\oplus \cdots \oplus H_{l+m}.
\end{equation} 

Let $\n_0=\left \{0\right\}\cup \n$ and we call $(l,m,k)\in \n_0^3$ an {\it admissible triple} if $k$ is one of $|l-m|$, $|l-m+2|$, $\cdots$, $\l+m$, i.e. $\displaystyle r=\frac{l+m-k}{2}\in \n_0$. For any admissible triple $(l,m,k)$ we have an associated isometry
\begin{equation}
\alpha^{l,m}_k:H_k\hookrightarrow H_l\otimes H_m
\end{equation}
and define quantum channels
\begin{align}
\Phi^{k\rightarrow l}_m:B(H_k)\rightarrow B(H_l),~\rho\mapsto (\text{id}\otimes \text{Tr})(\alpha^{l,m}_k\rho (\alpha^{l,m}_k)^*),
\end{align}
which we call $O_N^+$-Temperley-Lieb quantum channels. Their complementary quantum channels are given by 
\begin{align}
(\Phi^{k\rightarrow l}_m)^c:B(H_k)\rightarrow B(H_m),~\rho\mapsto (\text{Tr}\otimes \text{id})(\alpha^{l,m}_k\rho (\alpha^{l,m}_k)^*).
\end{align}

Since $H_n$ is isometrically embedded into $H_1^{\otimes n}$, we may assume that $\Phi^{k\rightarrow l}_m:B(H_k)\rightarrow B(H_1^{\otimes l})$ and $(\Phi^{k\rightarrow l}_m)^c:B(H_k)\rightarrow B(H_1^{\otimes m})$. Let us denote by $p_n:H_1^{\otimes n}\rightarrow H_n$ the orthogonal projection onto $H_n$. Then we can write the isometry $\alpha^{l,m}_k$ explicitly as
\begin{equation}\label{eq2}
\alpha^{l,m}_k= \sqrt{\frac{[k+1]_q}{\theta_q(k,l,m)}}(p_l\otimes p_m)(\text{id}_{l-r}\otimes T_{2r}\otimes \text{id}_{m-r}):H_k\rightarrow H_l\otimes H_m.
\end{equation}
Here, $T_{2r}:\Comp\rightarrow H_1^{\otimes 2r}$ is given by $1 \mapsto |i_1i_2\cdots i_r\ra\otimes |i_r\cdots i_2i_1\ra$. The {\it quantum integer} is given by $[0]_q=1$ and $[n+1]_q=\text{dim}(H_n)$. The {\it quantum factorial} is defined by $[n+1]_q!=[n+1]_q\cdot [n]_q\cdot \cdots \cdot [1]_q$ and the {\it theta-net} $\theta_q(k,l,m)$ is given by 
\begin{equation}
 \frac{[r]_q![l-r]_q![m-r]_q![k+r+1]_q!}{[l]_q![m]_q![k]_q!}.
\end{equation}

\section{Main results}

The Stinespring isometries $\alpha^{l,m}_k:H_k\hookrightarrow H_l\otimes H_m\hookrightarrow H_1^{\otimes l}\otimes H_1^{\otimes m}$ described in \eqref{eq2} seem complicated, but in turns out that they are approximated by much simpler isometries 
\begin{equation}
\gamma^{l,m}_k:\frac{1}{N^{\frac{r}{2}}}(\text{id}_{H_{l-r}}\otimes  T_{2r}\otimes \text{id}_{H_{m-r}}): H_k\rightarrow H_1^{\otimes l}\otimes H_1^{\otimes m}
\end{equation}
in the asymptotic regime.

\begin{theorem}\label{thm-main}
For any admissible triple $(l,m,k)\in \n_0^3$, we have
\begin{equation}
\norm{\alpha^{l,m}_k - \gamma^{l,m}_k }=O(N^{-1}).
\end{equation}
\end{theorem}

\begin{proof}
Since $\alpha^{l,m}_k$ is written as $\displaystyle \sqrt{\frac{[k+1]_qN^r}{\theta_q(l,m,k)}}(p_l\otimes p_m)\gamma^{l,m}_k$ on $H_k$, we have
\begin{equation}
\norm{(p_l\otimes p_m)\gamma^{l,m}_k | \xi\ra }=\sqrt{\frac{\theta_q(l,m,k)}{[k+1]_qN^r}}
\end{equation}
for any unit vector $\xi\in H_k$, so that
\begin{align*}
\norm{\gamma^{l,m}_k|\xi\ra-(p_l\otimes p_m)\gamma^{l,m}_k | \xi\ra }^2&=\norm{\gamma^{l,m}_k|\xi\ra }^2-\norm{ (p_l\otimes p_m)\gamma^{l,m}_k | \xi\ra }^2\\
&=1-\frac{\theta_q(l,m,k)}{[k+1]_qN^r}.
\end{align*}
In particular, we have
\begin{align*}
\norm{\gamma^{l,m}_k-(p_l\otimes p_m)\gamma^{l,m}_k}=\sqrt{1-\frac{\theta_q(l,m,k)}{[k+1]_qN^{r}}}=O(N^{-1}),
\end{align*}
by \cite[Lemma 4.2]{BCLY20}, and this leads us to the following estimates
\begin{align*}
&\norm{\alpha^{l,m}_k-\gamma^{l,m}_k}=\norm{\sqrt{\frac{[k+1]_q N^r}{\theta_q(l,m,k)}}(p_l\otimes p_m)\gamma^{l,m}_k-\gamma^{l,m}_k}\\
&\leq \left \| \left (\sqrt{\frac{[k+1]_qN^r}{\theta_q(l,m,k)} } -1\right )(p_l\otimes p_m)\gamma^{l,m}_k \right \|+\norm{(p_l\otimes p_m)\gamma^{l,m}_k-\gamma^{l,m}_k} \\
&\leq \left | \sqrt{\frac{[k+1]_qN^r}{\theta_q(l,m,k)} } -1 \right |+\norm{\gamma^{l,m}_k-(p_l\otimes p_m)\gamma^{l,m}_k}=O(N^{-1}).
\end{align*}

\end{proof}

The above Theorem \ref{thm-main} is enough to approximate $O_N^+$-TL quantum channels $\Phi^{k\rightarrow l}_m(\rho)=(\text{id}\otimes \text{Tr})(\alpha^{l,m}_k \rho (\alpha^{l,m}_k)^*)$ by
\begin{equation}
\Psi^{k\rightarrow l}_m(\rho)=(\text{id}\otimes \text{Tr})(\gamma^{l,m}_k \rho (\gamma^{l,m}_k)^*)=(\text{id}\otimes \text{Tr}_{H_1^{\otimes (m-r)}})(\rho)\otimes \frac{1}{N^r}\text{Id}_{H_1^{\otimes r}}
\end{equation}
in the sense that their {\it Bures distance} goes to zero as $N\rightarrow \infty$. Indeed, the Bures distance between two quantum channels $\Phi,\Psi:B(H_A)\rightarrow B(H_B)$ is measured by
\begin{equation}
\beta(\Phi,\Psi)=\inf  \left \|V_{\Phi}-V_{\Psi}\right \|
\end{equation}
where the infimum runs over all paris of Stinespring isometries $V_{\Phi},V_{\Psi}:H_A\rightarrow H_B\otimes H_E$ satisfying
\begin{equation}
\Phi(\rho)=(\text{id}\otimes \text{Tr}_E)(V_{\Phi}\rho V_{\Phi}^*)\text{ and }\Psi(\rho)=(\text{id}\otimes \text{Tr}_E)(V_{\Psi}\rho V_{\Psi}^*).
\end{equation}
Associated to this metric is so-called the {\it diamond distance} given by
\begin{equation}
\norm{\Phi-\Psi}_{\diamond}=\sup_{n\in \n} \norm{\text{id}_n\otimes (\Phi-\Psi)}_{S^1(\Comp^n \otimes H_A)\rightarrow S^1(\Comp^n \otimes H_B)},
\end{equation}
which satsisfies $\displaystyle \frac{1}{2}\norm{\Phi-\Psi}_{\diamond}\leq \beta(\Phi,\Psi)\leq \sqrt{\norm{\Phi-\Psi}_{\diamond}}$. See \cite{KSD08} for more details. Using these general facts, we obtain the following result for $O_N^+$-TL quantum channels as a corollary of Theorem \ref{thm-main}:

\begin{corollary}\label{cor2}
Let $(l,m,k)\in \n_0^3$ be an admissible triple and let $\Phi':B(H_{A'})\rightarrow B(H_{B'})$ be a quantum channel. Then, under notations from the above and Theorem \ref{thm-main}, we have
\begin{equation}
\beta(\Phi^{k\rightarrow l}_m\otimes \Phi',\Psi^{k\rightarrow l}_m\otimes \Phi')=O(N^{-1}).
\end{equation}
In particular, we have 
\begin{equation}
\norm{\Phi^{k\rightarrow l}_m\otimes \Phi'-\Psi^{k\rightarrow l}_m\otimes \Phi'}_{\diamond}=O(N^{-1}).
\end{equation}
\end{corollary}

A crucial fact is that a broad array of capacities are continuous with respect to the diamond distance $\norm{\cdot}_{\diamond}$ \cite{LS09}. Combining Corollaries 1,2,3 of \cite{LS09} and Corollary \ref{cor2}, we obtain 
\begin{equation}
\left |C(\Phi^{k\rightarrow l}_m\otimes \Phi') -C(\Psi^{k\rightarrow l}_m\otimes \Phi')\right |  = O\left (N^{-1}\log(N)\right ),
\end{equation}
and moreover the same conclusion holds for the quantum capacity $Q$ and the private classical capacity $P$. From now on, let's turn our attention to estimate capacities $C,Q,P$ of $\Psi^{k\rightarrow l}_m\otimes \Phi'$. The following proposition follows from an adaption of the proof of \cite[Theorem 4.5]{BCLY20}.

\begin{proposition}\label{prop1}
For any admissible triple $(l,m,k)$ and $N\geq 2$, we have
\small
\begin{equation}\label{eq3}
\frac{l+k-m}{2}\log(N-1)\leq Q^{(1)}(\Psi^{k\rightarrow l}_m)\leq C(\Psi^{k\rightarrow l}_m)\leq  \frac{l+k-m}{2}\log(N),
\end{equation}
\begin{equation}\label{eq4}
\small \frac{m+k-l}{2}\log(N-1)\leq Q^{(1)}((\Psi^{k\rightarrow l}_m)^c)\leq C((\Psi^{k\rightarrow l}_m)^c)\leq  \frac{m+k-l}{2}\log(N)
\end{equation}
\normalsize
\end{proposition}
\begin{proof}
For the first inequality of \eqref{eq3}, note that 
\begin{equation}
|i_1i_2\cdots i_{l-r}\ra\otimes |1212\cdots \ra \in H_k\subseteq H_1^{\otimes k}\
\end{equation}
for any $(i_1,i_2,\cdots, i_{l-r})\in \left \{1,2,\cdots,N\right\}^k$ such that $i_1\neq i_2$, $i_2\neq i_3$, $\cdots$, $i_{l-r-1}\neq i_{l-r}$ and $i_{l-r}\neq 1$. Let us denote by $S$ the set of all such $(N-1)^{l-r}$ vectors. Then for the following mixed state
\begin{equation}
\rho= \frac{1}{(N-1)^{l-r}}\sum_{\xi \in S}|\xi\ra\la \xi | \in B(H_{k})
\end{equation}
we have 
\begin{equation}
\left \{ \begin{array}{ll}
\displaystyle \Psi^{k\rightarrow l}_m(\rho)=\frac{1}{(N-1)^{l-r}}\sum_{\xi \in S}|i_1\cdots i_{l-r}\ra\la i_1\cdots i_{l-r} | \otimes \frac{1}{N^r}\text{Id}_{H_1}^{\otimes r}\\
\displaystyle (\Psi^{k\rightarrow l}_m)^c(\rho)=|1212\cdots \ra\la 1212\cdots |\otimes \frac{1}{N^r}\text{Id}_{H_1}^{\otimes r}.
\end{array} \right .
\end{equation}
This gives us
\begin{align*}
Q^{(1)}(\Psi^{k\rightarrow l}_m)&\geq H(\Psi^{k\rightarrow l}_m(\rho))-H((\Psi^{k\rightarrow l}_m)^c(\rho))\\
&=((l-r)\log(N-1)+r\log(N))-r\log(N)\\
&=(l-r)\log(N-1).
\end{align*}
The last inequality of \eqref{eq3} follows from the fact that
\begin{equation}
C(\Psi^{k\rightarrow l}_m)=C((\text{id}_{H_1^{\otimes (l-r)}}\otimes \text{Tr}_{H_1^{\otimes (m-r)}})(\cdot))\leq \log(N^{l-r}).
\end{equation}
On the other side, a similar proof works for \eqref{eq4} using vectors 
\begin{equation}
|\cdots 2121\ra\otimes |i_{l-r+1}i_{l-r+2}\cdots i_{k}\ra
\end{equation}
such that $1\neq i_{l-r+1}$, $i_{l-r+1}\neq i_{l-r+2}$, $\cdots$, $i_{k-1}\neq i_k$.
\end{proof}

\begin{remark}
Recall that \cite[Corollary 4.6]{BCLY20} states that for any admissible triple $(l,m,k)\in \n_0^3$ we have
\footnotesize
\begin{align}
& Q^{(1)}(\Phi^{k\rightarrow l}_m),P^{(1)}(\Phi^{k\rightarrow l}_m),\chi(\Phi^{k\rightarrow l}_m)=\frac{l+k-m}{2}\log(N)+O(N^{-1}\log(N)),\\
&Q^{(1)}((\Phi^{k\rightarrow l}_m)^c),P^{(1)}((\Phi^{k\rightarrow l}_m)^c),\chi((\Phi^{k\rightarrow l}_m)^c)=\frac{m+k-l}{2}\log(N)+O(N^{-1}\log(N)),
\end{align}
\normalsize but the questions for capacities $Q,P,C$ remained open. We can establish an affirmative answer on this by combining Proposition \ref{prop1} and Corollary \ref{cor2}: For any admissible triple $(l,m,k)\in \n_0^3$ we have
\small
\begin{align}
&Q(\Phi^{k\rightarrow l}_m),P(\Phi^{k\rightarrow l}_m),C(\Phi^{k\rightarrow l}_m)=\frac{l+k-m}{2}\log(N)+O(N^{-1}\log(N)),\\
&Q((\Phi^{k\rightarrow l}_m)^c),P((\Phi^{k\rightarrow l}_m)^c),C((\Phi^{k\rightarrow l}_m)^c)=\frac{m+k-l}{2}\log(N)+O(N^{-1}\log(N)).
\end{align}
\normalsize One more advantage is that the techniques in \cite{BCLY20} can be explained in a more straightforward way.
\end{remark}

An important fact is that our analysis works not only for additivity questions, but also for {\it strong additivity} questions for the capacities $C,P,Q$. Indeed, combining our asymptotic analysis and the fact that partial traces are strongly additive for capacities $C,P,Q$ \cite{GJL18a}, we can prove the following results for $O_N^+$-TL quantum channels in the asymptotic regime:

\begin{proposition}
Let $(l,m,k)\in \n_0^3$ be an admissible triple and let $\Phi':B(H_{A'})\rightarrow B(H_{B'})$ be a quantum channel. Then we have
\begin{equation}
C(\Psi^{k\rightarrow l}_m\otimes \Phi')=\frac{l+k-m}{2}\log(N)+C(\Phi')+O\left ( N^{-1}\right ).
\end{equation}
In particular, for $O_N^+$-TL quantum channels, we have 
\begin{align}
C(\Phi^{k\rightarrow l}_m\otimes \Phi')&=C(\Phi^{k\rightarrow l}_m)+C(\Phi')+O\left ( N^{-1}\log(N)\right ),\\
C((\Phi^{k\rightarrow l}_m)^c\otimes \Phi')&=C((\Phi^{k\rightarrow l}_m)^c)+C(\Phi')+O\left ( N^{-1}\log(N)\right )
\end{align}
Moreover, the same conclusion holds for the quantum capacity $Q$ and the private classical capacity $P$.
\end{proposition}
\begin{proof}
First of all, we have
\begin{align}
&C(\Psi^{k\rightarrow l}_m\otimes \Phi')\geq C(\Psi^{k\rightarrow l}_m)+ C( \Phi')\\
&\geq Q^{(1)}(\Psi^{k\rightarrow l}_m)+ C( \Phi')\geq \frac{l+k-m}{2}\log(N-1)+C(\Phi')
\end{align}
by superadditivity of $C$ and Proposition \ref{prop1}. The converse direction follows from the fact that $\Psi^{k\rightarrow l}_m$ is a restriction of a partial trace. More precisely, $B(H_k)$ is a subsystem of $B(H_1^{\otimes k})$, and the restriction of 
\[\widetilde{\Psi}(\rho)=(\text{id}_{H_1^{\otimes (l-r)}}\otimes \text{Tr}_{H_1^{\otimes (m-r)}})(\rho)\otimes \frac{1}{N^{r}}\text{Id}_{H_1^{\otimes r}}:B(H_1^{\otimes k})\rightarrow B(H_1^{\otimes l})\]
to $B(H_k)$ is exactly same with $\Phi^{k\rightarrow l}_m$. Thus, we have
\begin{equation}
C(\Psi^{k\rightarrow l}_m\otimes \Phi')\leq C(\widetilde{\Psi}\otimes \Phi')=\log(N^{l-r})+C(\Phi').
\end{equation}
Here, the last equality follows from strong additivity of direct sums of partial traces \cite[Proposition 5]{GJL18a}. These estimates tell us that 
\begin{equation}
C(\Psi^{k\rightarrow l}_m\otimes \Phi')=\frac{l+k-m}{2}\log(N)+C(\Phi')+O(N^{-1}),
\end{equation}
which leads us to following conclusion
\begin{align}
C(\Phi^{k\rightarrow l}_m\otimes \Phi')&=C(\Psi^{k\rightarrow l}_m\otimes \Phi')+O(N^{-1}\log(N))\\
&=\frac{l+k-m}{2}\log(N)+C(\Phi')+O(N^{-1}\log(N))\\
&=C(\Phi^{k\rightarrow l}_m)+C(\Phi')+O(N^{-1}\log(N)).
\end{align}
Here, the first and last equalities are thanks to Corollary \ref{cor2} and Proposition \ref{prop1} respectively. Similar arguments apply to the cases of the quantum capacity and the private classical capacity, and also for the complementary quantum channels $(\Phi^{k\rightarrow l}_m)^c$.
\end{proof}

\bibliographystyle{alpha}
\bibliography{Youn}

\end{document}